\documentclass[a4paper, 11pt]{article}
\usepackage[margin=1.2in]{geometry}

\usepackage{graphicx}
\usepackage[mathlines]{lineno}
\usepackage{amsmath}
\usepackage{amsfonts}
\usepackage{mathrsfs}
\usepackage{mathtools,amssymb}
\usepackage{amsthm}

\newtheorem{theorem}{Theorem}[section]

\newtheorem{definition}[theorem]{Definition}

\makeatletter
\newcommand\footnoteref[1]{\protected@xdef\@thefnmark{\ref{#1}}\@footnotemark}
\makeatother

\begin{document}
\title{On Information Links}

\author{Pierre Baudot\\ Median Technologies,1800 Route des Crêtes, 06560 Valbonne, France\\ pierre.baudot@gmail.com}

\maketitle              

\begin{abstract}
	\begin{sloppypar}
		In a joint work with D. Bennequin \cite{Baudot2019}, we suggested that the (negative) minima of the 3-way multivariate mutual information correspond to Borromean links, paving the way for providing probabilistic analogs of linking numbers. This short note generalizes the correspondence of the minima of $k$ multivariate interaction information with $k$ Brunnian links in the binary variable case. Following \cite{Jakulin2004}, the negativity of the associated K-L divergence of the joint probability law with its Kirkwood approximation implies an obstruction to local decomposition into lower order interactions than $k$, defining a local decomposition inconsistency that reverses  Abramsky's contextuality local-global relation \cite{Abramsky2011}. Those negative k-links provide a straightforward definition of collective emergence in complex k-body interacting systems or dataset.  
	\end{sloppypar}
	
\end{abstract}
\section{Introduction}
	\begin{sloppypar}
		Previous works established that Gibbs-Shannon entropy function $H_k$ can be characterized (uniquely up to the multiplicative constant of the logarithm basis) as the first class of a Topos cohomology defined on random variables complexes (realized as the poset of partitions of atomic probabilities), where marginalization corresponds to localization (allowing to construct Topos of information) and where the coboundary operator is an Hochschild's coboundary with a left action of conditioning (\cite{Baudot2015a,Vigneaux2019}, see also the related results found independently by Baez, Fritz and Leinster \cite{Baez2014,Baez2011}). Vigneaux could notably underline the correspondence of the formalism with the theory of contextuality developed by Abramsky \cite{Abramsky2011,Vigneaux2019}. Multivariate mutual informations $I_k$ appear in this context as coboundaries \cite{Baudot2015a,Baudot2019a}, and quantify refined and local statistical dependences in the sens that $n$ variables are mutually independent if and only if all the $I_k$ vanish (with $1<k<n$, giving $(2^n-n-1)$ functions), whereas the Total Correlations $G_k$ quantify global or total dependences, in the sens that $n$ variables are mutually independent if and only if  $G_n=0$ (theorem 2  \cite{Baudot2019}).  As preliminary uncovered by several related studies, information functions and statistical structures not only present some co-homological but also homotopical features that are finer invariants \cite{Baudot2019a,Bennequin2020,Manin2020}. Notably, proposition 9 in \cite{Baudot2019}, underlines a correspondence of the minima $I_3=-1$ of the mutual information between 3 binary variables with Borromean link. 
		For $k\geq3$, $I_k$ can be negative \cite{Hu1962}, a phenomenon called synergy and first encountered in neural coding \cite{Brenner2000}  and frustrated systems \cite{Matsuda2000} (cf. \cite{Baudot2019} for a review and examples of data and gene expression interpretation). However, the possible negativity of the $I_k$ has posed problems of interpretation, motivating a series of study to focus on non-negative decomposition, “unique information” \cite{Williams2010,Olbrich2015,Bertschinger2014}. Rosas et al. used such a positive decomposition to define emergence in multivariate dataset \cite{Rosas2020}. Following \cite{Baudot2019,Baudot2019a}, the present work promotes the correspondence of emergence phenomena with homotopy links and information negativity. 
		The chain rule of mutual-information goes together with the following inequalities discovered by Matsuda \cite{Matsuda2000}. For all random variables $X_{1};..;X_{k}$ with associated joint probability distribution $P$ we have:
		$I(X_{1};..;X_{k-1}|X_{k};P) \geq 0 ~ \Leftrightarrow ~ I_{k-1} \geq I_{k}$ and 
		$I(X_{1};..;X_{k-1}|X_{k};P) < 0  ~ \Leftrightarrow ~ I_{k-1} < I_{k}$  that characterize the phenomenon of information negativity as an increasing or decreasing sequence of mutual information. It means that  positive conditional mutual informations imply that the introduction of the new variable decreased the amount of dependence, while negative conditional mutual informations imply that the introduction of the new variable increase the amount of dependence. The meaning of conditioning, notably in terms of causality, as been studied at length in Bayesian network study, DAGs  (cf. Pearl's book \cite{Pearl1988}), although not in terms of information at the nice exception of Jakulin and Bratko \cite{Jakulin2004}. Following notably their work and the work of Galas and Sakhanenko \cite{Galas2019} and Peltre \cite{Peltre2020} on M\"{o}bius functions, we adopt the convention of Interaction Information functions $J_k=(-1)^{k+1}I_k$, which consists in changing the sign of even multivariate $I_k$ (remark 10 \cite{Baudot2019}). This sign trick, as called in topology, makes that  even and odd $J_k$ are both super-harmonic, a kind of pseudo-concavity in the sens of theorem D \cite{Baudot2015a}. In terms of Bayesian networks, interaction information (respectively conditional interaction information) negativity generalize the independence relation (conditional independence resp.), and identifies common consequence scheme (or multiple cause). Interaction information (Conditional resp.) negativity can be considered as an extended or "super" independence (Conditional resp.) property, and hence $J(X_{1};..;X_{n}|X_{k};P) < 0$ means that $(X_{1};..;X_{n})$ are $n$ "super" conditionally independent given $X_{k}$. Moreover such negativity captures common causes in DAGs and Bayesian networks. Notably, the cases where 3 variables are pairwise independent but yet present a (minimally negative) $J_3$ what is called the Xor problem in Bayesian network \cite{Jakulin2004}.  
		In \cite{Baudot2019}, we proposed  that those minima correspond to Borromean links. k-links are the simplest example of rich and complex families of link or knots (link and knot can be described equivalently by their braiding or braid word). k-links are prototypical homotopical invariants, formalized as link groups by Milnor \cite{Milnor1954}. A Brunnian link is a nontrivial link that becomes a set of trivial unlinked circles if any one component is removed. They are peculiarly attracting because of their beauty and apparent simplicity, and provide a clear-cut illustration of what is an emergent or purely collective property, they can only appear in 3-dimensional geometry and above (not in 1 or 2 dimension), and their complement in the 3-sphere were coined as "the mystery of three-manifolds" by Bill Thurston in the first slide of his presentation for Perelman proof of Poincar\'{e} conjecture (cf. Figure \ref{3-links}). In what follows, we will first generalize the correspondence of $J_k$ minima with k-links for arbitrary $k$ and then establish that negativity of interaction information detects incompatible probability, which allows an interpretation with respect to Abramsky's formalism as contextual. The information links are clearly in line with the principles of higher structures developed by \cite{Baas2019,Baas2012} that uses links to account for group interactions beyond pair interaction to catch the essence of many multi-agent interactions. Unravelling the formal relation of information links with the work of Baas, or with Khovanov homology \cite{Khovanov2006} are open questions.\\
	\end{sloppypar}

	\section{Information functions - definitions}
	
	\paragraph{Entropy.} the joint-entropy is defined  by \cite{Shannon1948} for any joint-product of $k$ random variables  $(X_1,..,X_k)$ with $\forall i \in [1,..,k], X_i\leq \Omega$ and for a probability joint-distribution  $\mathbb{P}_{(X_1,..,X_k)}$:
	\begin{equation} \label{jointentropy multiple}
	H_k = H(X_{1},..,X_{k};P) = c\sum_{x_1 \in [N_1],..,x_k \in  [N_k]}^{N_1\times..\times N_k}p(x_1,..,x_k)\ln p(x_1,..,x_k)\\
	\end{equation}
	where $[N_1\times...\times N_k]$ denotes the "alphabet" of $(X_1,...,X_k)$. More precisely, $H_k$ depends on 4 arguments: first, the sample space: a finite set $N_{\Omega}$; second a probability law $P$ on $N_{\Omega}$;  third, a set of random variable on $N_{\Omega}$, which is a surjective map $X_{j}:N_{\Omega} \rightarrow N_j$ and provides a partition of $N_{\Omega}$, indexed by the elements $x_{j_i}$ of $N_j$.  $X_j$ is less fine than $\Omega$, and write $X_{j}\leq \Omega$, or $\Omega \rightarrow X_{j}$, and the joint-variable $(X_i, X_j)$ is the less fine partition, which is finer than $ X_i $ and $ X_j $; fourth, the arbitrary constant $c$ taken here as $c = -1/\ln 2$. Adopting this more exhaustive notation, the entropy of $X_{j}$ for $P$ at $\Omega$ becomes $H_{\Omega}(X_{j};P)=H(X_{j};P_{X_j})=H(X_{j*}(P))$, where $X_{j*}(P)$ is the \emph{marginal} of $P$ by $X_{j}$ in $\Omega$.
	\paragraph{Kullback-Liebler divergence and cross entropy.} Kullback-Liebler divergence \cite{Kullback1951}, is defined for two probability laws $P$ and $Q$  having included support ($support(P)\subseteq support(Q)$) on a same probability space $(\Omega,\mathcal{F},P)$, noted $p(x)$ and $q(x)$ by: 
	\begin{equation}\label{relative entropy}
	D_{\Omega}(X_j;P,Q)= D(X_j;p(x)||q(x))=c\sum_{x \in \mathscr{X}}p(x)\ln \frac{q(x)}{p(x)} 
	\end{equation}
	
	\begin{sloppypar}
		\paragraph{Multivariate Mutual informations.} The k-mutual-information (also called co-information) are defined by \cite{McGill1954,Hu1962}:
		\begin{equation}\label{n-mutual information}
		I_k=I(X_{1};...;X_{k};P) = c\sum_{x_1,...,x_k\in [N_1\times...\times N_k]}^{N_1\times...\times N_k}p(x_1.....x_k)\ln \frac{\prod_{I\subset [k];card(I)=i;i \ \text{odd}} p_I}{\prod_{I\subset [k];card(I)=i;i \ \text{even}} p_I} 
		\end{equation}
		For example, $I_2=c\sum p(x_1,x_2)\ln \frac{p(x_1)p(x_2)}{p(x_1,x_2)}$ and  the 3-mutual information is the function $I_3=c\sum p(x_1,x_2,x_3)\ln\frac{p(x_1)p(x_2)p(x_3)p(x_1,x_2,x_3)}{p(x_1,x_2)p(x_1,x_3)p(x_2,x_3)}$. 
		We have the alternated sums or inclusion-exclusion rules \cite{Hu1962,Matsuda2000,Baudot2015a}:
		\begin{equation}\label{Alternated sums of information}
		I_n=I(X_1;...;X_n;P)=\sum_{i=1}^{n}(-1)^{i-1}\sum_{I\subset [n];card(I)=i}H_i(X_I;P)
		\end{equation}
		And the dual inclusion-exclusion relation (\cite{Baudot2019a}): 
		\begin{equation}\label{Alternated sums of entropy}
		H_n=H(X_1,...,X_n;P)=\sum_{i=1}^{n}(-1)^{i-1}\sum_{I\subset [n];card(I)=i}I_i(X_I;P)
		\end{equation}
		As noted by Matsuda \cite{Matsuda2000}, it is related  to the general Kirkwood superposition approximation, derived in order to approximate the statistical physic quantities and the distribution law by only considering the k first k-body interactions terms, with $k<n$ \cite{Kirkwood1935} :  
		\begin{equation}\label{Kirkwood1935}
		\hat{p}(x_1,...,x_n)= \frac{\prod_{I\subset [n];card(I)=n-1}p(x_I)}{\frac{\prod_{I\subset [n];card(I)=n-2}p(x_I)}{\frac{\colon}{\prod_{i=1}^np(x_i)}}}= \prod_{k=1}^{n-1}(-1)^{k-1}\prod_{I\subset [n];card(I)=n-k}p(x_I)
		\end{equation}
		For example : 
		\begin{equation}\label{Kirkwood3}
		\hat{p}(x_1,x_2,x_3)= \frac{p(x_1,x_2)p(x_1,x_3)p(x_2,x_3)}{p(x_1)p(x_2)p(x_3)}
		\end{equation}
		We have directly that $I_3 = D(p(x_1,...,x_n)||\hat{p}(x_1,...,x_n)) $ and $-I_4=D(p(x_1,...,x_n)||\hat{p}(x_1,...,x_n))$. It is hence helpful to introduce the "twin functions" of k-mutual Information called the \textbf{k-interaction information}  \cite{Jakulin2004,Baudot2019} (multiplied by minus one compared to \cite{Jakulin2004}), noted $J_n$:  
		\begin{equation}\label{interaction information}
		J_n=J(X_1;...;X_n;P)=(-1)^{i-1}I_n(X_1;...;X_n;P)
		\end{equation}
		Then we have  $J_n=D(p(x_1,...,x_n)||\hat{p}(x_1,...,x_n)) $. Hence, $J_n$ can be used to quantify how much the Kirkwood approximation of the probability distribution is "good" (in the sense that $p=\hat{p}$ if and only if $J_n=D(p||\hat{p})=0 $).The alternated sums or inclusion-exclusion rules becomes \cite{Jakulin2004}: 
		\begin{equation}\label{Alternated sums of interaction_information}
		J_n=J(X_1;...;X_n;P)=\sum_{i=1}^{n}(-1)^{n-i}\sum_{I\subset [n];card(I)=i}H_i(X_I;P)
		\end{equation}
		And now, their direct sums give the multivariate entropy, for example: $H_3=H(X_1,X_2,X_3)=J(X_1)+J(X_2)+J(X_3)+J(X_1;X_2)+J(X_1;X_3)+J(X_2;X_3)+J(X_1;X_2;X_3)$
		Depending on the context (\cite{Baudot2019} (p.15)) and the properties one wishes to use, one should use for convenience either $I_k$ or $J_k$ functions.  \\
		
		\paragraph{Conditional Mutual informations.} The conditional mutual information of two variables $X_{1};X_{2}$ knowing $X_3$, also noted $X_3.I(X_{1};X_{2})$, is defined as \cite{Shannon1948}: 
		\begin{equation}\label{conditional mutual information}
		I(X_{1};X_{2}|X_3;P)=c\sum_{x_1,x_2,x_3\in [N_1\times N_2\times N_3]}^{N_1\times N_2\times N_3} p(x_1,x_2,x_3)\ln \frac{p(x_1,x_3)p(x_2,x_3)}{p(x_3)p(x_1,x_2,x_3)}
		\end{equation} 
	\end{sloppypar}
	
	\section{Information k-links}
	
	\begin{figure}
		\centering
		\includegraphics[height=5cm]{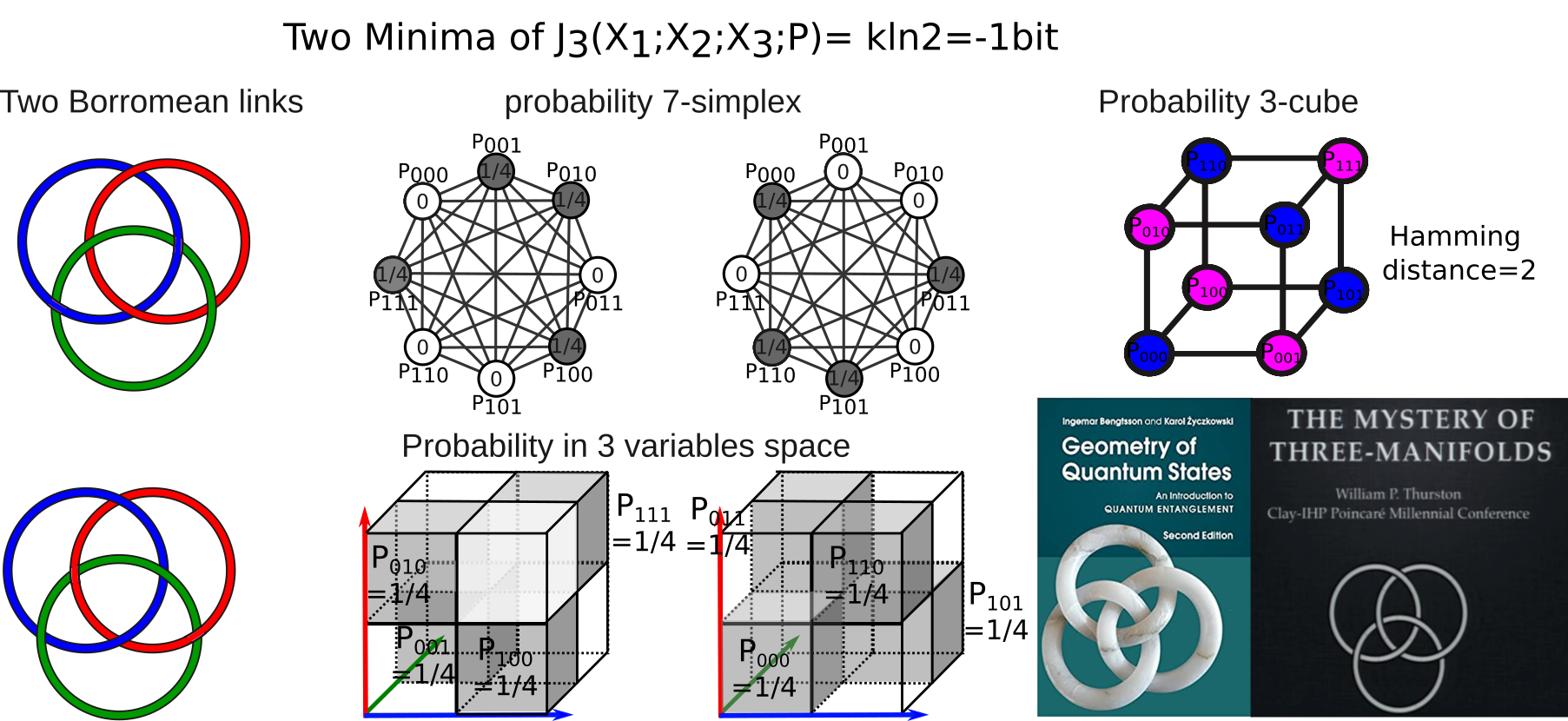}
		\caption{\textbf{Borromean 3-links of information (adapted from \cite{Baudot2019}).} from left to right. The two Borromean links (mirror images). The corresponding two probability laws in the probability 7-simplex for 3 binary random variables, and below the representation of the corresponding configuration in 3 dimensional data space. The variables are individually maximally entropic ($J_1=1$), fully pairwise independent  ($J_2=0$ for all pairs), but minimally linked by a negative information interaction ($J_3=-1$). The  corresponding graph covering in the 3-cube of the two configurations. The Hamming distance between two vertex having the same probability value is 2. The cover of Bengtsson and \.{Z}yczkowski's book on "the geometry of quantum states" and the first slide of the conference on "The Mystery of 3-Manifolds" by Bill Thurston at the Clay-IHP Milenial conference.}
		\label{3-links}
	\end{figure}
	\noindent Consider 3 binary random variables, then we have:
	\begin{sloppypar}
		\begin{theorem}[Borromean links of information \cite{Baudot2019}]
			The absolute minimum of $J_3$, equal to $-1$, is attained only in the two cases of three two by two independent unbiased variables satisfying $p_{000}=1/4,p_{001}=0, p_{010}=0, p_{011}=1/4, p_{100}=0, p_{101}=1/4, p_{110}=1/4, p_{111}=0$, or $p_{000}=0,p_{001}=1/4, p_{010}=1/4, p_{011}=0, p_{100}=1/4, p_{101}=0, p_{110}=0, p_{111}=1/4$. These cases correspond to the two borromean 3-links, the right one and the left one (cf. Figure \ref{3-links}, see \cite{Baudot2019} p.18).
		\end{theorem} 
		For those 2 minima, we have $H_1=1$, $H_2=2$, $H_3=2$, and $I_1=1$, $I_2=0$, $I_3=-1$, and $G_1=1$, $G_2=0$, $G_3=1$, and $J_1=1$, $J_2=0$, $J_3=-1$. 
		The same can be shown for arbitrary n-ary variables, which opens the question of the possiblity or not to classify more generaly others links. For example, for ternary variables, a Borromean link is achieved for the state $p_{000}=1/9,p_{110}=1/9,p_{220}=1/9,p_{011}=1/9,p_{121}=1/9,p_{201}=1/9,p_{022}=1/9,p_{102}=1/9,p_{212}=1/9$, and all others atomic probabilities are 0. The values of information functions are the same but in logarithmic basis 3 (trits) instead of 2 (bits). In other word, the probabilistic framework may open some new views on the classification of links. 
	\end{sloppypar}
	
	We now show the same result for the 4-Brunnian link, considering 4 binary random variables: 
	\begin{theorem}[4-links of information\footnote{\label{note1} As the proof relies on a weak concavity theorem D \cite{Baudot2015a} which proof has not been provided yet, this theorem shall be considered as a conjecture as long as the proof of theorem D \cite{Baudot2015a} is not given.}]
		The absolute minimum of $J_4$, equal to $-1$, is attained only in the two cases of four two by two and three by three independent unbiased variables satisfying 
		$p_{0000}=0,p_{0001}=1/8, p_{0010}=1/8, p_{0011}=0, p_{0101}=0, p_{1001}=0, p_{0111}=1/8, p_{1011}=1/8,p_{1111}=0,p_{1101}=1/8, p_{1110}=1/8, p_{0110}=0, p_{1010}=0, p_{1100}=0, p_{1000}=1/8, p_{0100}=1/8 $, or $p_{0000}=1/8,p_{0001}=0, p_{0010}=0, p_{0011}=1/8, p_{0101}=1/8, p_{1001}=1/8, p_{0111}=0, p_{1011}=0,p_{1111}=1/8,p_{1101}=0, p_{1110}=0, p_{0110}=1/8, p_{1010}=1/8, p_{1100}=1/8, p_{1000}=0, p_{0100}=0 $. These cases correspond to the two
		4-Brunnian links, the right one and the left one (cf. Figure \ref{4-links}).
	\end{theorem}

	\begin{figure}
		\centering
		\includegraphics[height=5cm]{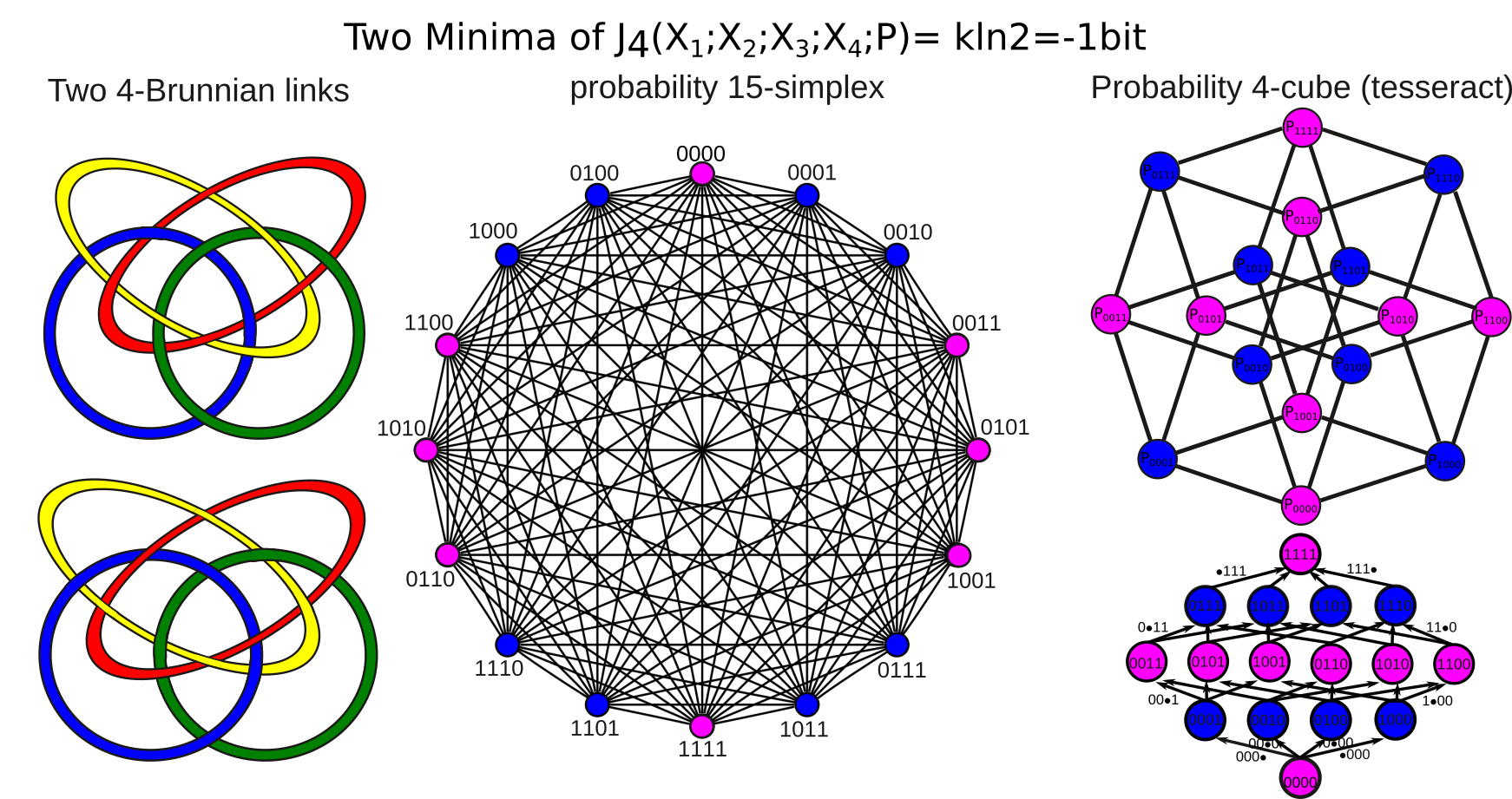}
		\caption{\textbf{4-links of information.} from left to right. The two 4-links (mirror images). The corresponding two probability laws in the probability 15-simplex for 4 binary random variables. The variables are individually maximally entropic ($J_1=1$), fully pairwise and tripletwise independent  ($J_2=0$, $J_3=0$, the link is said Brunnian), but minimally linked by a negative 4-information interaction ($J_4=-1$). The  corresponding graph covering in the 4-cube of the two configurations, called the tesseract. The Hamming distance between two vertex having the same probability value is 2. The corresponding lattice representation is illustrated bellow.}
		\label{4-links}
	\end{figure}

	\begin{proof} 
		\begin{sloppypar}
			The minima of $J_4$ are the maxima of $I_4$. We have easily $I_4\geq -\min(H(X_1),H(X_2),H(X_3),H(X_4))$ and  $I_4\leq \min(H(X_1),H(X_2),H(X_3),H(X_4))$ \cite{Matsuda2000}. Consider the case where all the variables are k-independent for all $k<4$ and all $H_1$ are maximal, then a simple combinatorial argument shows that $I_4= \binom{1}{4}.1-\binom{2}{4}.2+\binom{3}{4}.3+\binom{4}{4}.H_4$, which gives $I_4=4-H_4$. Now, since $I_4\leq \min(H(X_1),H(X_2),H(X_3),H(X_4))\leq \max(H(X_i)=1$, we have $H_4=4-1=3$  and $I_4=1$ or $J_4=-1$,  and it is a maxima of $I_4$ because it achieves the bound  $I_4\leq \min(H(X_1),H(X_2),H(X_3),H(X_4))\leq \max(H(X_i))=1$. To obtain the atomic probability values and see that there are two such maxima, let's consider all the constraint imposed by independence. We note the 16 unknown:  
		\end{sloppypar}
		
		\begin{align*}
		\scriptstyle	a=p_{0000},\quad b=p_{0011},\quad c=p_{0101},\quad d=p_{0111},\quad e=p_{1001},\quad f=p_{1011},\quad  g=p_{1101},\quad h=p_{1110}, \\
		\scriptstyle  \quad i=p_{0001},\quad j=p_{0010},\quad k=p_{0100},\quad l=p_{0110},\quad m=p_{1000},\quad n=p_{1010},\quad o=p_{1100},\quad p=p_{1111}.
		\end{align*}
		
		The maximum entropy (or 1-independence) of single variable gives 8 equations:
		\begin{linenomath*}
			\begin{gather*}
			\scriptstyle	a+b+c+d+i+j+k+l=1/2,\quad e+f+g+h+m+n+o+p=1/2,\quad a+b+e+f+i+j+m+n=1/2, \\
			\scriptstyle	c+d+g+h+k+l+o+p=1/2,\quad  a+c+e+g+i+k+m+o=1/2,\quad  b+d+f+h+j+l+n+p=1/2,\\
			\scriptstyle  a+h+j+k+l+m+n+o=1/2,\quad b+c+d+e+f+g+i+p=1/2.	
			\end{gather*}
		\end{linenomath*}
		
		The 2-independence of (pair of) variables gives the 24 equations:
		\begin{linenomath*}
			\begin{gather*}
			\scriptstyle	a+b+i+j=1/4,\quad c+d+k+l=1/4,\quad e+f+m+n=1/4,\quad g+h+o+p=1/4,\quad a+c+i+k=1/4,\\
			\scriptstyle	b+d+j+l=1/4,\quad e+g+m+o=1/4,\quad f+h+n+p=1/4,\quad a+j+k+l=1/4,\quad b+c+d+i=1/4,\\
			\scriptstyle	h+m+n+o=1/4,\quad e+f+g+p=1/4,\quad a+e+i+m=1/4,\quad b+f+j+n=1/4,\quad c+g+k+o=1/4,\\
			\scriptstyle	d+h+l+p=1/4,\quad a+j+m+n=1/4,\quad b+e+f+i=1/4,\quad h+k+l+o=1/4,\quad c+d+g+p=1/4,\\
			\scriptstyle    a+k+m+o=1/4,\quad c+e+g+i=1/4,\quad h+j+l+n=1/4,\quad b+d+f+p=1/4.
			\end{gather*}
		\end{linenomath*}
		
		The 3-independence of (triplet of) variables gives the 32 equations:
		
		\begin{linenomath*}
			\begin{gather*}
			\scriptstyle	a+j=1/8,\: b+i=1/8,\: k+l=1/8,\: c+d=1/8,\: m+n=1/8,\: e+f=1/8,\: h+o=1/8,\: g+p=1/8,\\
			\scriptstyle	a+i=1/8,\: b+j=1/8,\: c+h=1/8,\: d+l=1/8,\: e+m=1/8,\: f+n=1/8,\: vg+o=1/8,\: h+p=1/8,\\
			\scriptstyle	a+m=1/8,\: e+i=1/8,\: n+j=1/8,\: b+f=1/8,\: k+o=1/8,\: c+g=1/8,\: h+l=1/8,\: d+p=1/8,\\
			\scriptstyle	a+k=1/8,\: c+i=1/8,\: j+l=1/8,\: b+d=1/8,\: m+o=1/8,\: e+g=1/8,\: k+n=1/8,\: f+p=1/8.
			\end{gather*}
		\end{linenomath*}
		
		\noindent Since $J_4$ is super-harmonic (weakly concave, theorem D \cite{Baudot2015a})\footnoteref{note1} which implies that minima of $J$ happen on the boundary of the probability simplex, we have one  additional constraint that for example $a$ is either $0$ or $1/8$. Solving this system of 64 equations with a computer with $a=0$ or $a=1/(2^3)$ gives the two announced solutions. Alternatively, one can remark that out of the 64 equations only $\sum_{k=1}^4 \binom{k}{4}-2+1=2^4-2+1=15$ are independent with  $a=0$ or $a=1/(2^3)$, the 2 systems are hence fully determined and we have 2 solutions. Alternatively, it could be possible to derive a geometric proof using the 4-cube as covering graph (called tesseract), of probability simplex, establishing that $0$ probabilities only connects $1/8$ probabilities as illustrated in Figure \ref{4-links}.
	\end{proof}

	\noindent For those 2 minima, we have $H_1=1$, $H_2=2$, $H_3=3$, $H_4=3$, and $I_1=1$, $I_2=0$, $I_3=0$, $I_4=1$, and  $G_1=1$, $G_2=0$, $G_3=0$, $G_4=1$, $J_1=1$,$J_2=0$, $J_3=0$, $J_3=-1$.
	
	\noindent The preceding results generalizes to $k$-Brunnian link: consider $k$ binary random variables then we have the theorem: 
	\begin{theorem}[k-links of information\footnoteref{note1}]
		The absolute minimum of $J_k$, equal to $-1$, is attained only in the two cases j-independent j-uplets of unbiased variables for all $1<j<k$ with atomic probabilities $p(x_1,...,x_k) = 1/2^{k-1}$ or $p(x_1,...,x_k) = 0$ such that the associated vertex of the associated k-hypercube covering graph of $p(x_1,...,x_k) = 1/2^{k-1}$ connects only vertices of $p(x_1,...,x_k) = 0$ and conversely. These cases correspond to the two k-Brunnian links, the right one and the left one.
	\end{theorem}     
	\begin{proof}  
		\begin{sloppypar}
			Following the same line as previously. We have $J_k\geq -\min(H(X_1)...,H(X_k))$ and  $I_k\leq \min(H(X_1),...,H(X_k))$ .
			Consider the case where all the variables are i-independent for all $i<k$ and all $H_1$ are maximal, then a simple combinatorial argument shows that $I_k=\sum_{i=1}^{k-1} (-1)^{i-1} \binom{i}{k}.i+(-1)^{k-1}H(X_1,...,X_k)$ that is $I_k=k-H(X_1,...,X_k)$  now since $I(X_1;...;X_k)\leq \min(H(X_1),...,H(X_k))\leq \max(H(X_i)=1$, we have $H(X_1,...,X_k)=k-1$  and $I_k=(-1)^{k-1}$ or $J_4=-1$,  and it is a minima of $J_k$ because it saturates the bound  $J(X_1;...;X_k)\geq \max(H(X_i))= -1$.  i-independent for all $i<k$ imposes a system of $2^k-2+1=15$ independent equations (the +1 is for $\sum p_i=1$). Since $J_k$ is weakly concave (theorem D \cite{Baudot2015a})\footnoteref{note1} which is equivalent to say that minima of $J$ happen on the boundary of the probability simplex, we have one  additional constraint that for example $a$ is either $0$ or $1/2^{k-1}$. It gives 2 systems of equations that are hence fully determined and we have 2 solutions. The probability configurations corresponding to those two solutions can  be found by considering the k-cube as covering covering graph (well known to be bipartite: it can be colored with only two colors) of the probability $2^k-1$-simplex, establishing that $0$ probabilities only connects $1/2^{k-1}$ probabilities. 
		\end{sloppypar}
	\end{proof} 
	
	For those 2 minima, we have: for $-1<i<k$ $H_i=i$ and $H_k=k-1$. We have
	$I_1=1$, for $1<i<k$ $I_i=0$, and $I_k={-1}^{k-1}$. We have  $G_1=1$, for $1<i<k$ $G_i=0$, and $G_k=1$. We have  $J_1=1$,for $1<i<k$ $J_i=0$,and $J_k=-1$.
	
	\section{Negativity and Kirkwood decomposition inconsistency}
	
	\begin{sloppypar} 
		The negativity of K-L divergence can happen in certain cases, in an extended context of measure theory. A measure space is a probability space that does not necessarily realize the axiom of probability of a total probability equal to 1 \cite{Kolmogorov1933a}, e.g $P(\Omega)=\sum q_i = k$, where k is an arbitrary real number (for real measure).  In the seminal work of Abramsky and Brandenburger \cite{Abramsky2011},  two probability laws $P$ and $Q$ are contextual whenever there does not exist any joint probability that would correspond to such marginals (but only non-positive measure). In such cases, measures are said incompatible, leading to some obstruction to the existence of a global section probability joint-distribution:  $P$ and $Q$ are locally consistent but globally inconsistent. This section underlines that interaction information negativity displays a kind of  "dual" phenomena, that we call inconsistent decomposition (or indecomposability), whenever interaction information is negative on a consistent global probability law then no local Kirkwood decomposition can be consistent, leading to an obstruction to decomposability (e.g. globally consistent but locally inconsistent decomposition measure). 
	\end{sloppypar}	
	
	\begin{definition} [Inconsistent decomposition] 
		A measure space $P$ is consistent whenever $P(\Omega)=\sum p_i = 1$ (and hence a probability space), and inconsistent otherwise.
	\end{definition}
	We first show that given a probability law $P$, if $D(P,Q)<0$ then $Q$ is inconsistent and cannot be a probability law. 
	\begin{theorem} [K-L divergence negativity and inconsistency] 
		Consider a probability space with probability $P$ and a measure space with measure $Q$, the negativity of the Kullback-Leibler divergence $D_{\Omega}(X_j;P,Q)$ implies that $Q$ is inconsistent.
	\end{theorem}
	
	\begin{proof}
		\begin{sloppypar}
		The proof essentially relies on basic argument of the proof of convexity of K-L divergence or on Gibbs inequalities. Consider the K-L divergence between $P$ and $Q$:
		$D_{\Omega}(X_j;P,Q)= D(X_j;p(x)||q(x))=c\sum_{i \in \mathscr{X}}p_i\ln \frac{q_i}{p_i}$.
		Since $\forall x > 0, \ \ln x \leqslant x-1$  with equality if and only if $x=1$, hence we have with $c=-1/\ln 2$; $ k\sum_{i \in \mathscr{X}}p_i\ln \frac{q_i}{p_i} \geq  c\sum_{i \in \mathscr{X}} p_i(\frac{q_i}{p_i}-1)$. We have $c\sum_{i \in \mathscr{X}} p_i(\frac{q_i}{p_i}-1) = c \left( \sum_{i \in \mathscr{X}} q_i- \sum_{i \in \mathscr{X}} p_i\right)$.  $P$ is a probability law, then by the axiom 4 of probability \cite{Kolmogorov1933a}, we have $\sum q_i \neq 1$  and hence $c \left( \sum_{i \in \mathscr{X}} q_i- \sum_{i \in \mathscr{X}} p_i\right)= c \left( \sum_{i \in \mathscr{X}} q_i -1\right)$. Hence if $D_{\Omega}(X_j;P,Q)<0$ then  $\sum_{i \in \mathscr{X}} q_i > 1$, and since $\sum_{i \in \mathscr{X}} q_i > 1$, by definition $Q$ is inconsistent.
		\end{sloppypar}
	\end{proof}
	
	\begin{sloppypar}
		\begin{theorem} 
			[Interaction negativity and inconsistent Kirkwood decomposition, adapted from \cite{Jakulin2004} p.13] for $n>2$, if $J_n<0$ then no Kirkwood probability decomposition subspace $P_{X_K}$ defined by the variable products of $(X_K)$ variables with $K\subset [n]$ is consistent.
		\end{theorem}
	\end{sloppypar}
	
	\begin{proof} 
		\begin{sloppypar}
			The theorem is proved by remarking, that interaction negativity corresponds precisely to cases where the Kirkwood approximation is not possible (fails) and would imply a probability space with $\sum p_i>1$ which contradicts the axioms of probability theory. We will use a proof by contradiction. Let's assume that the probability law follows the Kirkwood approximation which can be obtained from $n-1$ products of variables and we have \ref{Kirkwood1935}: 
			\begin{equation}
			\hat{p}(x_1,...,x_n)= \frac{\prod_{I\subset [n];card(I)=n-1}p(x_I)}{\frac{\prod_{I\subset [n];card(I)=n-2}p(x_I)}{\frac{\colon}{\prod_{i=1}^np(x_i)}}}
			\end{equation}
			Now consider that $J_n<0$ then, since $-\log x < 0$ if and only if $x>1$, we have   $\left(\frac{\prod_{I\subset [n];card(I)=i;i \ \text{odd}} p_I}{\prod_{I\subset [n];card(I)=i;i \ \text{even}} p_I}\right)^{(-1)^{n-1}}>1$, which is the same as   $\frac{\hat{p}(x_1,...,x_n)}{p(x_1,...,x_n)}>1$ and hence $\hat{p}(x_1,...,x_n)>p(x_1,...,x_n)$. Then summing over all atomic probabilities, we obtain
			$\sum_{x_1 \in [N_1],...,x_n \in  [N_n]}^{N_1\times...\times N_n}\hat{p}(x_1,...,x_n)>\sum_{x_1 \in [N_1],...,x_n \in  [N_n]}^{N_1\times...\times N_n}p(x_1,...,x_n)$ and hence $\sum_{x_1 \in [N_1],...,x_n \in  [N_n]}^{N_1\times...\times N_n}\hat{p}(x_1,...,x_n)>1$ which contradicts axiom 4 of probability \cite{Kolmogorov1933a}. Hence there does not exist probability law on only $n-1$ products of variables satisfying $J_n<0$. If there does not exist probability law on only $n-1$ products of variables, then by marginalization on the $n-1$ products, there does not exist such probability law for all $2<k<n$ marginal distributions on $k$ variable product. 
		\end{sloppypar}
	\end{proof}
	
	Remark: Negativity of interaction information provides intuitive insight into contextual interactions as obstruction to decomposition-factorization into lower order interactions, which is classical here in the sense that it does not rely on quantum formalism: quantum information extends this phenomenon to self or pairwise interactions: Information negativity happens also for $n=2$ or $n=1$ precisely for the states that violate Bell's inequalities \cite{Cerf1997,Horodecki2007}.\\
	
	\textbf{Acknowledgments:} I thank warmly anonymous reviewer for helpful remarks improving the manuscript and Daniel Bennequin whose ideas are at the origin of this work.

%
%
%
\bibliographystyle{splncs04}
\bibliography{bibtopo}
%
	
	




\end{document}